\documentclass[smallextended]{svjour3-arxiv}     % onecolumn (second format); required for SCW, see http://www.springer.com/economics/economic+theory/journal/355
\smartqed  % flush right qed marks, e.g. at end of proof
\usepackage{graphicx}
\usepackage{wrapfig}
\newcommand\eat[1]{}
\usepackage{tikz}
\usepackage{booktabs}  %% nice tables
\usetikzlibrary{arrows}
\usepackage{slashbox}
\journalname{Working Paper}
\usepackage{array,xspace,multirow,hhline,graphicx,xcolor,tikz,colortbl,tabularx,amsmath,amssymb,amsfonts}
\usepackage[round]{natbib}
\usetikzlibrary{shapes}
\usepackage{algorithm,algorithmic}
\usepackage{mathrsfs}
\usepackage{enumitem}
\setenumerate[1]{label=(\emph{\roman{*}}),ref=(\emph{\roman{*}}),leftmargin=*}

\usepackage{txfonts}
\usepackage{eucal} % consistency with minstable

\usepackage{tikz}

\newlength{\wordlength}

	\newcommand{\ceil}[1]{\lceil #1 \rceil }
	\newcommand{\PS}{{Multi-unit-eating $PS$}\xspace}
	\newcommand{\ps}{{multi-unit-eating $PS$}\xspace}

\newcommand{\pref}{\succsim\xspace}
\newcommand{\Pref}[1][]{
	\ifthenelse{\equal{#1}{}}{\mathrel \succsim}{\mathop{\succsim_{#1}}}
}                                          
\newcommand{\sPref}[1][]{                  
	\ifthenelse{\equal{#1}{}}{\mathrel \succ}{\mathop{\succ_{#1}}}
}                                          
\newcommand{\Indiff}[1][]{                 
	\ifthenelse{\equal{#1}{}}{\mathrel \sim}{\mathop{\sim_{#1}}}
}
\newcommand{\prefset}[1][]{\ifthenelse{\equal{#1}{}}{\mathcal{\succsim}}{\mathcal{\succsim}_{#1}}}

\newcommand{\ml}[1][]{\ensuremath{\ifthenelse{\equal{#1}{}}{\mathit{ML}}{\mathit{ML}(#1)}}}
\newcommand{\sml}[1][]{\ensuremath{\ifthenelse{\equal{#1}{}}{\mathit{SML}}{\mathit{SML}(#1)}}}
\newcommand{\sd}[1][]{\ensuremath{\ifthenelse{\equal{#1}{}}{\mathit{SD}}{\mathit{SD}(#1)}}}
\newcommand{\rsd}[1][]{\ensuremath{\ifthenelse{\equal{#1}{}}{\mathit{RSD}}{\mathit{RSD}(#1)}}}
\newcommand{\rd}[1][]{\ensuremath{\ifthenelse{\equal{#1}{}}{\mathit{RD}}{\mathit{RD}(#1)}}}
\newcommand{\st}[1][]{\ensuremath{\ifthenelse{\equal{#1}{}}{\mathit{ST}}{\mathit{ST}(#1)}}}
\newcommand{\bd}[1][]{\ensuremath{\ifthenelse{\equal{#1}{}}{\mathit{BD}}{\mathit{BD}(#1)}}}
\newcommand{\pc}[1][]{\ensuremath{\ifthenelse{\equal{#1}{}}{\mathit{PC}}{\mathit{PC}(#1)}}}
\newcommand{\dl}[1][]{\ensuremath{\ifthenelse{\equal{#1}{}}{\mathit{DL}}{\mathit{DL}(#1)}}}
\newcommand{\ul}[1][]{\ensuremath{\ifthenelse{\equal{#1}{}}{\mathit{UL}}{\mathit{UL}(#1)}}}

\newcommand{\set}[1]{\{#1\}}
\newcommand{\midd}{\mathbin{: }}

\let\enumtemp=\enumerate
\def\enumerate{\enumtemp\itemsep 1pt}
\let\itemtemp=\itemize
\def\itemize{\itemtemp\itemsep 1pt}

\newcommand{\Omit}[1]{}

	\newtheorem{observation}{Observation}%

\begin{document}

\title{Random assignment with multi-unit demands}
	% \subtitle{Do you have a subtitle?\\ If so, write it here}


	\author{Haris Aziz}


	\institute{%
	  Haris Aziz \at
	  NICTA and UNSW
	  2033 Sydney, Australia \\
	  Tel.: +61-2-8306\,0490 \\
	  Fax: +61-2-8306\,0405 \\
	  \email{haris.aziz@nicta.com.au}
	}

	\date{Received: date / Accepted: date}

\maketitle

\begin{abstract}
	We consider the multi-unit random assignment problem in which agents express preferences over objects and objects are allocated to agents randomly based on the preferences.
	%, and there may be more objects than agents. 
	%In randomized settings, agents need to reason about their random allocations. 
The most well-established preference relation to compare random allocations of objects is stochastic dominance ($\sd$) which also leads to corresponding notions of envy-freeness, efficiency, and weak strategyproofness.
We show that there exists no rule that is anonymous, neutral, efficient and weak strategyproof. For single-unit random assignment, we show that there exists no rule that is anonymous, neutral, efficient and weak group-strategyproof.
%One of the results carries over to the setting of randomized voting.
		% We present a general impossibility result for the random assignment setting with multi-unit demands namely that there exists no anonymous, neutral, weak $SD$-strategyproof, and $SD$-efficient rule. Similarly, there exists no anonymous, neutral, weak group$SD$-strategyproof, and $SD$-efficient rule even if the number of objects is not more than the number of agents. The latter result also carries over to the setting of randomized voting. 
		% For the single-unit demand random assignment problem in which each agent is allocated at most one object, the \textit{probabilistic serial mechanism ($PS$)} has been shown to be envy-free, weak strategyproof, and efficient. For multiple-unit demand, $PS$ is not weak strategyproof.In this paper, we study another generalization of $PS$ for multi-unit demand called \textit{multi-unit-eating $PS$} which was defined by Che and Kojima (2010).
We then study a generalization of the $PS$ (probabilistic serial) rule called multi-unit-eating $PS$ and prove that multi-unit-eating $PS$ satisfies envy-freeness, weak strategyproofness, and unanimity. 
\end{abstract}
	


\section{Introduction}

In the assignment problem, agents express linear preferences over objects and an object is assigned to each agent keeping in view the agents' preferences. The problem models one of the most fundamental setting in computer science and economics with numerous applications~\citep{Gard73b,Wils77a,Youn95b,Sven94a,Sven99a,BEL10a,ACMM05a}. Depending on the application setting, the objects could be car-park spaces, dormitory rooms, replacement kidneys, school seats, etc.
The assignment problem is also referred to as \emph{house allocation}~\citep{ACMM05a,AbSo99a}. 
If the outcome of the assignment problem is \emph{deterministic} then it can be inherently unfair. Take the example of two agents having identical preferences over two objects. Then any reasonable notion of fairness demands that both agents have equal right to each of the two objects. Since randomization is one of the oldest tools to achieve fairness, we consider the \emph{random assignment problem}~\citep{HyZe79a,Youn95b,BoMo01a,KaSe06a,GuCo10a,BCK11a,BCKM12a} in which objects are allocated randomly to agents according to their preferences.
The outcome is a random assignment which specifies the probability of each object being allocated to each of the agents. In contrast to some of the earlier work on random assignment, we focus on the random assignment problem in which there can be more objects than the number of agents~\citep{Koji09a}.

When agents express ordinal preferences over objects but the outcomes are fractional or randomized allocations, then there is a need to use \emph{lottery extensions} to extend preferences over objects to preferences over random allocations. 
In random settings, the most established preference relation between random allocations is \emph{stochastic dominance ($\sd$)}. $\sd$ requires that one random allocation is preferred to another one if and only if the former first-order stochastically dominates the latter. This relation is especially important because one random allocation stochastically dominates another one  if and only if the former yields at least as much expected utility as the latter for any von-Neumann-Morgenstern (vNM) utility representation consistent with the ordinal preferences~\citep{ABS13a}.
The $\sd$ relation can be used to define corresponding notions of envy-free, efficiency, and strategyproofness~\citep{BoMo01a,KaSe06a}. In this paper, we check
which levels of fairness, efficiency, and strategyproofness can be satisfied simultaneously.

For the random assignment problem without multi-unit demands, the most common and well-known way to assign objects is \emph{random priority ($RP$)} in which a permutation of agents is chosen uniformly at random and agents successively take their most preferred available object~\citep{AbSo98a,BoMo01a,CrMo01a}. Although $RP$ is strategyproof and results in a Pareto optimal assignment, \citet{BoMo01a} in a remarkable paper showed that $RP$ does not satisfy the stronger efficiency notion of stochastic dominance ($\sd$) efficiency and also a fairness concept called $\sd$-envy-freeness.\footnote{Another drawback of $RP$ is that the resultant fractional allocation is \#P-complete to compute~\citep{ABB13b}.} 
Furthermore, they presented an elegant algorithm called \emph{$PS$ (probabilistic serial)} that is not only $\sd$-efficient and $\sd$-envy-free but also satisfies weak $\sd$-strategyproofness. In $PS$, agents `eat' the most favoured available object at the same rate until all the objects are consumed. The fraction of object consumed by an agent is the probability of the agent getting that object.\footnote{By the \emph{Birkhoff-von Neumann theorem}, any fractional assignment can be represented by a convex combination over discrete assignments.}

Since its inception~\citep{BoMo01a}, $PS$ has received considerable attention and has been extended in a number of ways~\citep{KaSe06a,AtSe11a,Yilm09a}. In particular, it can be naturally extended to the more general case with multi-unit demands in which there are $nc$ objects and $c\geq 1$ objects are allocated to each of the agents~\citep{BoMo01a,Heo11a,Koji09a}. The extension does not require any modification to the specification of $PS$:
agents continue eating their most preferred available object until all the objects have been consumed. Although this \emph{one-at-a-time} extension (which we will refer to as $OPS$) still satisfies $\sd$-efficiency and $\sd$-envy-freeness, it is not weak $\sd$-strategyproof~\citep{Koji09a}. Incidentally there is another extension of $PS$ called the \textit{multi-unit-eating probabilistic serial} that was briefly described by \citet{ChKo10a} but has received no attention in the literature. In multi-unit-eating $PS$, each agent tries to eat his $c$ most preferred objects that are still available at a uniform speed until all objects have been consumed. We show that multi-unit-eating $PS$ satisfies desirable properties: it is weak $\sd$-strategyproof, $\sd$-envy-free, and unanimous. 

We point out that the problem of discrete assignment with multi-unit demands has attracted considerable attention~\citep{BoLa11a,Budi11a,EhKl03a,Hatf09a,KNWX13a,BEL10a}. In this paper, we focus on \emph{random} assignments with multi-unit demands. Multi-unit demand is a natural requirement in settings such as course allocation~\citep{Budi11a}. Moreover, we will require that each agents gets equal number of objects~\citep{Hatf09a}. This is a natural requirement in settings such as paper assignment to referees. 
%available objects with speed $1$ at every time point until all the objects have been consumed. 

%Let $rem(t)$ be the number of objects that have not been completely eaten at time $t$. In multi-unit-eating $PS$, each agent eats his $\max(c,rem(t))$ favorite
%available objects with speed $1$ at every time point until all the objects have been consumed. 

%\paragraph{Related work}

Apart from $RP$ and $PS$, two other natural assignment rules are \emph{uniform} and \emph{priority}. In the uniform rule, each agent gets $1/n$ of each object~\citep{Cham04a,Koji09a}. In the priority mechanism, there is a permutation of agents, and each agent in the permutation is assigned the $c$ most preferred available objects. The priority mechanism is also referred to as \emph{serial dictator} in the literature~\citep{Sven94a,Sven99a}.
Whereas uniform does not take into account the preferences of agents and is highly inefficient, priority is highly unfair to the agents at the end of the permutation. In more recent work, \citet{NPV15a} proposed two mechanisms for the random assignment problem that also handle limited complementarities. \citet{Hash13a} presented a generalization of RP for more general settings. 
% \textbf{The introduction and the discussion on related work should focus more on other impossibility results. (for example, Papai 2001, Ehlers and Klaus 2003 and Hatfield 2009).

% \citep{Hatf09a}
% \citep{EhKl03a}
% \citep{Papa01b}
%\citep{Budi11a}
% \citep{NPV15a}
% \citep{Hash13a}

\paragraph{Contributions}

We first  prove that for multi-unit demands, there exists no anonymous, neutral, weak $\sd$-strategyproof and $\sd$-efficient random assignment rule. The statement is somewhat surprising considering that all the four axioms used in the statement are minimal requirements. Incidentally, we have not used $\sd$ envy-freeness that is often used to obtain characterizations or impossibility statements in the literature~\citep{Heo11a,BoMo01a,EhKl03a,Koji09a} and is a very demanding requirement. The result is then extended to random assignment \emph{without} multi-unit demands if 
requiring weak $\sd$ group-strategyproofness instead of weak $\sd$ strategyproofness.  Our second result carries over to the setting randomized voting in which agents express weak orders over alternatives and the outcome is a lottery over the alternatives. %For both the impossibility results, we show that if even one axiom is dropped, the impossibility does not hold. 
% and $(ii)$ randomized voting.
%Our statement carries over to various resource allocation domains with additive utilities.

We then conduct an axiomatic analysis of the \ps. 
It is first highlighted that the definition of \ps in the literature is not entirely correct. A proper definition of \ps is formulated.
We show that for multi-unit demands, in contrast to $OPS$, \ps satisfies weak $\sd$-strategyproofness. We prove that \ps satisfies $\sd$ envy-freeness which is one of the strongest notions of fairness.
On the other hand, \ps does not fare well in terms of efficiency. We prove that \ps does not even satisfy 
ex post efficiency although it does satisfy unanimity. Therefore when we generalize $PS$ for multi-unit demands, $OPS$ is the right extension if the focus is on efficiency. On the other hand \ps is the right extension, if the aim is to maintain weak $\sd$-strategyproofness. The arguments for weak $\sd$-strategyproofness and $\sd$ envy-freeness of $MPS$ \ps also simplify the proofs for $PS$ for single-unit demands in \citep{BoMo01a}.
The study helps clarify the relative merits of different assignment rules for multi-unit demands. The relative merits of prominent random assignment rules are then summarized in Table~\ref{table:summarymulti} in the final section.

\section{Preliminaries}

\paragraph{Random assignment problem}
The model we consider is the random assignment problem which is a triple $(N,O,\pref)$ where $N$ is the set of $n$ agents $\{1,\ldots, n\}$, $O=\{o_1,\ldots, o_m\}$ is the set of objects, and $\pref=(\pref_1,\ldots,\pref_n)$ specifies strict, complete, and transitive preferences $\pref_i$ of agent $i$ over $O$. 
We will assume that $m$ is a multiple of $n$ i.e., $m=nc$ where $c$ is an integer.
We will denote by $\mathcal{R}(O)$ as the set of all complete and transitive relations over the set of objects $O$.

% For all $i\in N$ and $o_j\in O$, a random assignment $p$ specifies $p(i)(o_j)\in \[0,1\]$ as the probability of agent $i$ getting allocated object $o_j$.
% For a random assignment $p$, $\sum_{i\in N}p(i)(o_j)= 1$ for all $j\in \{1,\ldots, n\}$; and $\sum_{o\in O}p(i)(o_j)= c$ for all $i\in N$. %$\sum_{h_j\in H}p_{ij}=\sum_{h\in H}w(h)/n$;
% We denote by $p(i)$ the allocation of agent $i$. 
% 
% 
% The value $p(i)(o_j)$ is the fraction of object $o_j$ that agent $i$ is allocated. Each row $p(i)=(p(i)(o_1),\ldots, p(i)(o_m))$ represents the allocation of agent $i$. 
% The set of columns correspond to the objects $o_1,\ldots, o_m$.
% A feasible random assignment is discrete if $p(i)(o)\in \{0,1\}$ for all $i\in N$ and $o\in O$.

A random assignment $p$ is a $(n\times m)$ matrix $[p(i)(o_j)]_{1\leq i\leq n, 1\leq j\leq m}$ such that for all $i\in N$, and $o_j\in O$, $ p(i)(o_j) \in [0,1]$; $\sum_{i\in N}p(i)(o_j)= 1$ for all $j\in \{1,\ldots, m\}$; and $\sum_{o_j\in O}p(i)(o_j)= c$ for all $i\in N$. %$\sum_{h_j\in H}p_{ij}=\sum_{h\in H}w(h)/n$;
The value $p(i)(o_j)$ represents the probability of object $o_j$ being allocated to  agent $i$. Each row $p(i)=(p(i)(o_1),\ldots, p(i)(o_m))$ represents the allocation of agent $i$. 
The set of columns correspond to probability vectors of the objects $o_1,\ldots, o_m$.
A feasible random assignment is discrete if $p(i)(o)\in \{0,1\}$ for all $i\in N$ and $o\in O$.
A \emph{random assignment rule} specifies for each preferences profile a random assignment.
Two minimal fairness conditions for rules are \emph{anonymity} and \emph{neutrality}. Informally, they require that the rule should not depend on the names of the agents or objects respectively. % A condition that is even weaker than anonymity is that of equal treatment of equals (ETE) which requires that the allocation of any two agents with identical preferences should be the same. We say that a random assignment rule is symmetric if it satisfies ETE.

% The allocations for all the agents can be described by an assignment matrix, with the rows indexing the agents and col- umns indexing the houses; like the endowment matrix, the assignment matrix will be a doubly sub-stochastic matrix.

We define the \emph{SD (stochastic dominance)} relation which is an incomplete relation that extends the preferences of the agents over objects to preferences over random allocations. 
Given two random assignments $p$ and $q$, $p(i) \succsim_i^{\sd} q(i)$ i.e.,  a player $i$ \emph{$\sd$~prefers} allocation $p(i)$ to allocation $q(i)$ if 
	$	\sum_{o_j\in \set{o_k\midd o_k\succsim_i o}}p(i)(o_j) \ge \sum_{o_j\in \set{o_k\midd o_k\succsim_i o}}q{(i)(o_j)} \text{ for all } o\in O.$ Since SD is incomplete, it can be that two allocations $p(i)$ and $q(i)$ are \emph{incomparable}: $p(i) \not\succsim_i^{\sd} q(i)$ and  $q(i) \not\succsim_i^{\sd} p(i)$.

	Next, we define the DL (downward lexicographical) relation which is a complete relation. 
	Let $p(i)$ and $q(i)$ be two random allocations. Let $o\in O$ be the most preferred object such that $p(i)(o)\neq q(i)(o)$. Then,
	$p(i) \succ_i^{DL} q(i) \iff p(i)(o)> q(i)(o).$

\begin{example}\label{example:assignment}
	Consider the random assignment problem for two agents $N=\{1,2\}$ and four objects $O=\{o_1,o_2,o_3,o_4\}$ with the following preferences:
	\begin{align*}
		1: o_1,o_2,o_3,o_4\\
		2: o_2,o_1,o_3,o_4\\
	\end{align*}
Let us assume that agent $1$ gets $o_1$ with probability one, and objects $o_3$ and $o_4$ with probability half. Then the random assignment can be represented by the following matrix.
	\[p=\begin{pmatrix}
			1&0&1/2&1/2\\
		  0&1&1/2&1/2
			\end{pmatrix}.
		\]Note that agent 1's preference is $o_1\mathrel{\succ_1}o_2\mathrel{\succ_1}o_3\mathrel{\succ_1}o_4$. Based on the preferences over objects, one can consider preferences over allocations: $p(1)\succ_1^{\sd} p(2)$ and also $p(1)\succ_1^{DL} p(2).$
	% Let us assume that agent one has the following preferences $o_1\mathrel{\succ_1}o_2\mathrel{\succ_1}o_3\mathrel{\succ_1}o_4$ which can also be represented as $1: o_1,o_2,o_3,o_4$.
	% 	% \begin{align*}
	% 	% 	1: o_1,o_2,o_3,o_4.
	% 	% \end{align*}
	%  Then, $p(1)\succ_1^{\sd} p(2)$ and also $p(1)\succ_1^{DL} p(2).$
\end{example}

\paragraph{Envy-freeness}

	An assignment $p$ satisfies \emph{$\sd$ envy-freeness} if each agent (weakly) $\sd$ prefers his allocation to that of any other agent: $p(i) \pref_i^{\sd} p(j) \text{ for all } i,j\in N.$
	An assignment $p$ satisfies \emph{weak $\sd$ envy-freeness} if no agent strictly $\sd$ prefers someone else's allocation to his: $\neg[p(j)\succ_i^{\sd} p(i)] \text{ for all } i,j\in N.$
	For fairness concepts, $\sd$ envy-freeness implies weak $\sd$-envy-freeness~\citep{BoMo01a}.

\paragraph{Economic efficiency}

	An assignment is \emph{perfect} if each agents gets his most preferred $c$ objects. 
	% A discrete assignment $p$ is \emph{Pareto optimal} if there does not exists another discrete assignment $q$ such that $q(i) \succsim_i^{\sd} p(i)$ for all $i\in N$ and $q(i) \succ_i^{\sd} p(i)$ for some $i\in N$.
	An assignment $p$ is \emph{$\sd$-efficient} is there exists no assignment $q$ such that $q(i) \succsim_i^{\sd} p(i)$ for all $i\in N$ and $q(i) \succ_i^{\sd} p(i)$ for some $i\in N$. An assignment is \emph{ex post efficient} if it can be represented as a probability distribution over the set of $\sd$-efficient discrete assignments. Perfection implies $\sd$-efficiency which implies ex post efficiency.
	
An assignment rule is $\sd$-efficient (ex post efficient) if it always returns an $\sd$-efficient (ex post efficient) assignment. An assignment rule satisfies \emph{unanimity}, if it returns the perfect assignment if a perfect assignment exists.

$\sd$-efficiency implies ex post efficiency which implies unanimity.	The first implication was shown by \citep{BoMo01a}. For the second implication, assume that an assignment does not satisfy unanimity, there exists a perfect assignment $p$ but the mechanism returns some imperfect assignment $q$. The only $\sd$-efficient assignment that gives $c$ units to each agent is $p$. However since $q\ne p$, it cannot be achieved by a probability distribution over $\sd$-efficient discrete assignments.

%A notion weaker than SD-efficiency but stronger than ex post efficiency is that each decomposition is over Pareto optimal assigmments.	

	\paragraph{Strategyproofness}
	
	A random assignment function $f$ is \emph{$\sd$-strategyproof} if 
	$f(\pref)(i)\succsim_i^{\sd} f(\pref_i',\pref_{-i})(i)\text{ for all $\pref_i'$ and $\pref_{-i}$}.$
	A random assignment function $f$ is \emph{weak $\sd$-strategyproof} if 
	$\neg[f(\pref_i',\pref_{-i})(i) \succ_i^{\sd}  f(\pref)(i)]  \text{ for all $\pref_i'\in \mathcal{R}(O)$ and $\pref_i'\in {\mathcal{R}(O)}^{n-1}$}.$
It is easy to see that  $\sd$-strategyproofness implies weak $\sd$-strategyproofness~\citep{BoMo01a}.
A random assignment function $f$ is \emph{weak $\sd$-group-strategyproof} if there never exists an $S\subset N$ and $\pref_S'\in {\mathcal{R}(O)}^{|S|}$ such that 
$f(\pref_{S}',\pref_{-S})(i) \succ_i^{\sd}  f(\pref)(i) \text{ for all } i\in S$ and $\pref_{-S}\in {\mathcal{R}(O)}^{n-|S|}$.

% $\sd$-efficiency implies discrete $\sd$-efficiency implies ex post efficiency which implies unanimity.

% \begin{proposition}
% 	$\sd$-efficiency implies discrete $\sd$-efficiency implies ex post efficiency which implies unanimity.
% \end{proposition}

% For fairness concepts, $\sd$ envy-freeness implies weak $\sd$-envy-freeness~\citep{BoMo01a}.% whereas anonymity is logically incomparable to the envy-freeness notions.  		

% In the \textit{one-at-a time probabilistic serial} mechanism, each agent eats the most preferred available object with speed 1 at every time $t\in [0,k]$.

% 
% If there exists an assignment in which each agents gets his most preferred allocation, then the mechanism returns such an assignment.
% 
% 
% 
% 
% ``Discrete assignment is Pareto optimal'': If the mechanism returns a discrete assignment, it is Pareto optimal. Ex post efficiency implies ``Discrete assignment is Pareto optimal''.

%AGMN+13a

%\section{Related work}

\section{General impossibilities}

For the random assignment problem for which the number of objects is not more than the number of agents, there exists a rule ($PS$) that is anonymous, neutral, $\sd$-efficient and weak $\sd$-strategyproof. However when the number of objects is more than the number of agents, we get the following impossibility (Theorem~\ref{th:impossible}). 

%\footnote{The theorem requires an even weaker property than anonymity called \emph{equal treatment of equals} 
%The impossibility statement concerns a very basic setting and involves undemanding axioms. Hence, it applies to a variety of more complex resource allocation settings.

\begin{theorem}\label{th:impossible}
For the random assignment problem with $c>1$, there exists no anonymous, neutral, $\sd$-efficient, and 
weak $\sd$-strategyproof rule.
\end{theorem}
\begin{proof}
	
	We consider a random assignment setting with two agents and four objects with the requirement that each agents gets two units of houses.
\begin{align*}
\pref_1:&\quad a,b,c,d\\
\pref_2:&\quad  b,c,a,d\\
\pref_1':&\quad b,a,c,d\\
\pref_2':&\quad  b,a,c,d
\end{align*}

Let us compute $f(\pref_1,\pref_2')$. %By SD-efficiency, anonymity, and neutrality, we know that.
By anonymity and neutrality of $f$

\[f(\pref_1,\pref_2')=\begin{pmatrix}
		w&x&y&z\\
	  x&w&y&z
		\end{pmatrix}.
	\]

By $\sd$-efficiency of $f$,

\[f(\pref_1,\pref_2')=\begin{pmatrix}
		1&0&y&z\\
	  0&1&y&z
		\end{pmatrix}.
	\]

By anonymity and neutrality of $f$, 

\[f(\pref_1,\pref_2')=\begin{pmatrix}
		1&0&1/2&1/2\\
	  0&1&1/2&1/2
		\end{pmatrix}.
	\]

By using similar arguments, $\sd$-efficiency, anonymity, and neutrality of $f$ implies that

\[f(\pref_1',\pref_2)=\begin{pmatrix}
		1&1/2&0&1/2\\
	  0&1/2&1&1/2
		\end{pmatrix}.
	\]

Now let us consider

\[f(\pref_1,\pref_2)=\begin{pmatrix}
		x_{11}&	x_{12}&	x_{13}&	x_{14}\\
	  	x_{21}&	x_{22}&	x_{23}&	x_{24}
		\end{pmatrix}.
	\]

	For $f(\pref_1,\pref_2)$ to be feasible,
	
\begin{align*}
x_{11},	x_{12},	x_{13},	x_{14},	x_{21},	x_{22},	x_{23},	x_{24}\geq 0\\
x_{11}+	x_{12}+	x_{13}+	x_{14}=2\\
x_{21}+	x_{22}+	x_{23}+	x_{24}=2\\
x_{11}+	x_{21}=x_{12}+	x_{22}= x_{13}+	x_{23}= x_{14}+	x_{24}=1\\
\end{align*}

% 	$f(\pref_1,\pref_2)(2)(b)=1$ because if it were not then DL-SP is not satisfied.
% 	Hence $f(\pref_1,\pref_2)(1)(b)=0$.
% 	$f(\pref_1,\pref_2)(1)(a)=1$ because if it were not then DL-SP is not satisfied.
% 	But now agent 1 can improve his allocation by expressing $\pref_1$ to get $1/2$ of $b$ rather than $0$ of $b$. Hence $f$ cannot be DL-strategyproof.  
% 	
% 	
% Now if we want to avoid weak $\sd$-SP, either allocations are indifferent or there are incomparabilities. Indifference does not do the job so we try incomparabilities.

% The first thing to notice is that since $f(\pref_1,\pref_2)$ is $\sd$-efficient, then either $x_{21}=0$ or $x_{13}=0$. If not, then agent $2$ can give $\epsilon$ fraction of $a$ to agent $1$ in return for  $\epsilon$ fraction of $c$ and both agents benefit. This implies that $f(\pref_1,\pref_2)$ is not $\sd$-efficient.

Next, we show that if $f(\pref_1,\pref_2)= f(\pref_1',\pref_2)$ or $f(\pref_1,\pref_2)= f(\pref_1,\pref_2')$, then $f$ is not weak $\sd$-strategyproof. 

If $f(\pref_1,\pref_2)= f(\pref_1',\pref_2)$, 
then 

\[f(\pref_1,\pref_2')(2)\succ_2^{\sd} f(\pref_1,\pref_2)(2).\]
Hence, $f$ is not weak $\sd$-strategyproof.  

If $f(\pref_1,\pref_2)= f(\pref_1,\pref_2')$,
then 

\[f(\pref_1',\pref_2)(1)\succ_1^{\sd} f(\pref_1,\pref_2)(1).\]
Hence, $f$ is not weak $\sd$-strategyproof.  

% If $f(\pref_1,\pref_2)= f(\pref_1,\pref_2')$, then $f$ is not weak SD-strategyproof. 
Therefore the only way $f$ can still be weak $\sd$-strategyproof if both of the following conditions hold.

\begin{itemize}
\item $f(\pref_1,\pref_2)(1)$ is incomparable for $1$ with 
$f(\pref_1',\pref_2)(1)$.	
\item $f(\pref_1,\pref_2)(2)$ is incomparable for $2$ with 
$f(\pref_1,\pref_2')(2)$.
\end{itemize}

This means that the following constraints should hold.

Given that agent $2$ reports $\pref_2$, agent $1$ should not benefit by misreporting $\pref_1'$ instead of $\pref_1$. This implies that $x_{11}+x_{12}+x_{13}>1.5$.

Given that agent $1$ reports $\pref_1$, agent $2$ should not benefit by misreporting $\pref_2'$ instead of $\pref_2$. This implies that $x_{22}+x_{23}+x_{21}>1.5$.
 
%
% \begin{enumerate}
% \item \label{cond1}
%
% \begin{itemize}
% \item[] $x_{11}+x_{12}>1.5$ or $x_{11}+x_{12}+x_{13}>1.5$
% %\item $x_{11}+x_{12}+x_{13}>1.5$
% \end{itemize}
% and
% \item \label{cond2}
%  \begin{itemize}
% \item[] $1/2>x_{12}$ or $1.5>x_{11}+x_{12}$
% %\item[] $1.5>x_{11}+x_{12}$
% \end{itemize} and
% \item \label{cond3}
% \begin{itemize}
% \item[] $x_{22}+x_{23}>1.5$ or $x_{22}+x_{23}+x_{21}>1.5$
% %\item $x_{22}+x_{23}+x_{21}>1.5$
% \end{itemize}
% and
% \item \label{cond4}
% \begin{itemize}
% \item[] $1>x_{22}$ or $1.5> x_{21}+x_{22}+x_{23}$
% %\item[] $1.5> x_{21}+x_{22}+x_{23}$
% \end{itemize}
% % $f(\pref_1,\pref_2)(1)$ is incomparable for $1$ with
% % $f(\pref_1',\pref_2)(1)$.
% % \item $f(\pref_1,\pref_2)(2)$ is incomparable for $2$ with
% % $f(\pref_1,\pref_2')(2)$.
% \end{enumerate}
%
% Since each of the conditions \ref{cond1}, \ref{cond2}, \ref{cond3}, \ref{cond4} need to be satisfied, it means that the following inequalities hold:
% \begin{align*}
% x_{11}+x_{12}+x_{13}>1.5\\
% x_{22}+x_{23}+x_{21}>1.5
% \end{align*}

Adding both these inequalities yields 
\begin{align*}
x_{11}+x_{12}+x_{13} + x_{22}+x_{23}+x_{21}>3.
\end{align*}

But this is a contradiction since 
$x_{11}+x_{12}+x_{13} + x_{22}+x_{23}+x_{21}= (x_{11}+ x_{21}) + (x_{12}+ x_{22}) + (x_{13}+ x_{23})=3$. Hence if $f$ is $\sd$-efficient, and anonymous, neutral, then it cannot be weak $\sd$-strategyproof.

The same argument can be extended to arbitrary number of agents where each agent requires two objects from among $o_1,\ldots, o_{2n}$. Each new agent $i\in \{3,\ldots, n\}$ most prefers objects $o_{2i-1}, o_{2i}$ and least prefers objects $o_1,o_2,o_3,o_4$. Hence in each $\sd$-efficient assignment each agent $i\in \{3,\ldots, n\}$ is allocated $o_{2i-1}$ and $o_{2i}$ completely. The same arguments for the case of two agents apply to the more general case.
Similarly, the same arguments can also be extended to the case where $c>2$. One can add more objects to end of the preference lists of both agents and each agent gets a uniform fraction of these objects at the end of the preference lists. 
\qed
\end{proof}

Theorem~\ref{th:impossible} complements an earlier impossibility result of \citet{Koji09a} that states there exists no $\sd$-efficient, $\sd$ envy-free, and weak $\sd$-strategyproof random assignment rule for multi-unit demands. 
In Theorem~\ref{th:impossible}, the property of $\sd$ envy-freeness is replaced by anonymity.

The proof above can be extended by cloning agents $1$ and $2$ to prove the following statement for the basic assignment setting with single-unit demand.

\begin{theorem}\label{th:impossible-group}
For the random assignment problem, there exists no anonymous, neutral, $\sd$-efficient, and 
weak $\sd$ group-strategyproofness rule even for equal number of agents and objects.
\end{theorem}
\begin{proof}

	We consider a random assignment setting with four agents and fours objects. There are two agents that are of type $1$ and two agents of type $2$. Let the real preferences of the agents $\{1,2\}$ of type $1$ be $\pref_1$ and let the real preferences of agents $\{3,4\}$ of type $2$ be $\pref_2$.

\begin{align*}
\pref_1:&\quad a,b,c,d\\
\pref_2:&\quad  b,c,a,d\\
\pref_1':&\quad b,a,c,d\\
\pref_2':&\quad  b,a,c,d\\
\end{align*}

Let us compute $f(\pref_1,\pref_1,\pref_2',\pref_2')$. %By SD-efficiency, anonymity, and neutrality, we know that.

By anonymity and neutrality, we know that

\[f(\pref_1,\pref_1,\pref_2',\pref_2')=\begin{pmatrix}
		w/2&x/2&y/2&z/2\\
		w/2&x/2&y/2&z/2\\
	  x/2&w/2&y/2&z/2\\
	  x/2&w/2&y/2&z/2
		\end{pmatrix}.
	\]

By $\sd$-efficiency, we know that

\[f(\pref_1,\pref_1,\pref_2',\pref_2')=\begin{pmatrix}
		1/2&0&y/2&z/2\\
			1/2&0&y/2&z/2\\
	  0&1/2&y/2&z/2\\
			  0&1/2&y/2&z/2
		\end{pmatrix}.
	\]

Due to anonymity and neutrality of $f$, 

\[f(\pref_1,\pref_1,\pref_2',\pref_2')=\begin{pmatrix}
		1/2&0&1/4&1/4\\
		1/2&0&1/4&1/4\\
	  0&1/2&1/4&1/4\\
		 0&1/2&1/4&1/4
		\end{pmatrix}.
	\]

By using similar arguments, $\sd$-efficiency, anonymity, and neutrality of $f$ implies that

\[f(\pref_1',\pref_1',\pref_2,\pref_2)=\begin{pmatrix}
		1/2&1/4&0&1/4\\
		1/2&1/4&0&1/4\\
	  0&1/4&1/2&1/4\\
			  0&1/4&1/2&1/4
		\end{pmatrix}.
	\]

Now let us consider

\[f(\pref_1,\pref_1,\pref_2,\pref_2)=\begin{pmatrix}
		x_{11}/2&	x_{12}/2&	x_{13}/2&	x_{14}/2\\
		x_{11}/2&	x_{12}/2&	x_{13}/2&	x_{14}/2\\
	  	x_{21}/2&	x_{22}/2&	x_{23}/2&	x_{24}/2\\
			x_{21}/2&	x_{22}/2&	x_{23}/2&	x_{24}/2
		\end{pmatrix}.
	\]

	For $f(\pref_1,\pref_1, \pref_2,\pref_2)$ to be feasible,

\begin{align*}
x_{11},	x_{12},	x_{13},	x_{14},	x_{21},	x_{22},	x_{23},	x_{24}\geq 0\\
x_{11}+	x_{12}+	x_{13}+	x_{14}=2\\
x_{21}+	x_{22}+	x_{23}+	x_{24}=2\\
x_{11}+	x_{21}=x_{12}+	x_{22}= x_{13}+	x_{23}= x_{14}+	x_{24}=1\\
\end{align*}

% 	$f(\pref_1,\pref_2)(2)(b)=1$ because if it were not then DL-SP is not satisfied.
% 	Hence $f(\pref_1,\pref_2)(1)(b)=0$.
% 	$f(\pref_1,\pref_2)(1)(a)=1$ because if it were not then DL-SP is not satisfied.
% 	But now agent 1 can improve his allocation by expressing $\pref_1$ to get $1/2$ of $b$ rather than $0$ of $b$. Hence $f$ cannot be DL-strategyproof.  
% 	
% 	
% Now if we want to avoid weak $\sd$-SP, either allocations are indifferent or there are incomparabilities. Indifference does not do the job so we try incomparabilities.

% The first thing to notice is that since $f(\pref_1,\pref_1, \pref_2,\pref_2)$ is $\sd$-efficient, then either $x_{21}=0$ or $x_{13}=0$. If not, then agents $2$ can give $\epsilon$ fraction of $a$ to agent $1$ in return for  $\epsilon$ fraction of $c$ and all agents benefit. This implies that $f(\pref_1,\pref_1, \pref_2,\pref_2)$ is not $\sd$-efficient.

Next, we show that if $f(\pref_1,\pref_1,\pref_2,\pref_2)= f(\pref_1',\pref_1',\pref_2,\pref_2)$ or $f(\pref_1,\pref_1,\pref_2,\pref_2)= f(\pref_1,\pref_1,\pref_2',\pref_2')$, then $f$ is not  weak $\sd$ group-strategyproof. 

If $f(\pref_1,\pref_1,\pref_2,\pref_2)= f(\pref_1',\pref_1',\pref_2,\pref_2)$, 
then 
\[f(\pref_1,\pref_1,\pref_2',\pref_2')(3)\succ_2^{\sd} f(\pref_1,\pref_1,\pref_2,\pref_2)(3).\]
Hence, $f$ is not weak $\sd$ group-strategyproof.  

If $f(\pref_1,\pref_1,\pref_2,\pref_2)= f(\pref_1,\pref_1,\pref_2',\pref_2')$,
then 
\[f(\pref_1',\pref_1',\pref_2,\pref_2)(1)\succ_1^{\sd} f(\pref_1,\pref_1,\pref_2,\pref_2)(1).\]
Hence, $f$ is not weak $\sd$ group-strategyproof.  

% If $f(\pref_1,\pref_2)= f(\pref_1,\pref_2')$, then $f$ is not weak $\sd$-strategyproof. 

Given that agent of type $\pref_2$ report $\pref_2$, then agent of type $\pref_1$ should not benefit by misreporting $\pref_1'$ instead of $\pref_1$. This implies that $x_{11}+x_{12}+x_{13}>1.5$.

Given that agents of type $2$ report $\pref_1$, then agents of type $2$ should not benefit by misreporting $\pref_2'$ instead of $\pref_2$. This implies that $x_{22}+x_{23}+x_{21}>1.5$.

%
% Therefore the only way $f$ can still be weak $\sd$ group-strategyproof if both of the following conditions hold.
%
% \begin{itemize}
% \item $f(\pref_1,\pref_1,\pref_2,\pref_2)(1)$ is incomparable for type $1$ with
% $f(\pref_1',\pref_1',\pref_2\pref_2)(1)$.
% \item $f(\pref_1,\pref_1,\pref_2,\pref_2)(3)$ is incomparable for type $2$ with
% $f(\pref_1,\pref_1,\pref_2',\pref_2')(3)$.
% \end{itemize}
%
% This means that the following constraints should hold.
%
%
% \begin{enumerate}
% \item \label{cond1}
%
% \begin{itemize}
% \item $x_{11}+x_{12}>1.5$ or
% \item $x_{11}+x_{12}+x_{13}>1.5$
% \end{itemize}
% % and
% % \item \label{cond2}
% %  \begin{itemize}
% % \item[] $1>x_{11}$ or
% % \item[] $1.5>x_{11}+x_{12}$
% % \end{itemize}
% and
% \item \label{cond3}
% \begin{itemize}
% \item $x_{22}+x_{23}>1.5$ or
% \item $x_{22}+x_{23}+x_{21}>1.5$
% \end{itemize}
% % and
% % \item \label{cond4}
% % \begin{itemize}
% % \item[] $1>x_{22}$ or
% % \item[] $1.5> x_{21}+x_{22}+x_{23}$.
% % \end{itemize}
% % $f(\pref_1,\pref_2)(1)$ is incomparable for $1$ with
% % $f(\pref_1',\pref_2)(1)$.
% % \item $f(\pref_1,\pref_2)(2)$ is incomparable for $2$ with
% % $f(\pref_1,\pref_2')(2)$.
% \end{enumerate}
%
% %Since each of the conditions \ref{cond1}, \ref{cond2}, \ref{cond3}, \ref{cond4} need to be satisfied,
%
% This means that the following inequalities hold:
%
% \begin{align*}
% x_{11}+x_{12}+x_{13}>1.5\\
% x_{22}+x_{23}+x_{21}>1.5
% \end{align*}

Hence, 
\begin{align*}
x_{11}+x_{12}+x_{13} + x_{22}+x_{23}+x_{21}>3
\end{align*}

But this is a contradiction since 
$x_{11}+x_{12}+x_{13} + x_{22}+x_{23}+x_{21}= (x_{11}+ x_{21}) + (x_{12}+ x_{22}) + (x_{13}+ x_{23})=3$. Hence if $f$ is $\sd$-efficient and anonymous, and  neutral, then it cannot be weak $\sd$ group-strategyproof.

The same argument can be extended to arbitrary number of agents.
 %where each agent requires two objects from among $o_1,\ldots, o_{2n}$. Each new agent $i\in \{2,\ldots, n\}$ most prefers objects $o_{2i-1}, o_{2i}$ and least prefers objects $o_1,o_2,o_3,o_4$. Hence in each $\sd$-efficient assignment each agent $i\in \{2,\ldots, n\}$ is allocated $o_{2i-1}$ and $o_{2i}$ completely. The same arguments for the case of two agents apply to the more general case.
	\qed
\end{proof}

Theorem~\ref{th:impossible-group} (that holds for single-unit demands) complements Theorem 1 in \citep{Koji09a} that only holds for multi-unit demands. 
The assignment problem in which $m=n$ can be viewed as a subdomain of voting in which each alternative is a discrete assignment and preferences of an agent over assignments simply depend on his allocated object~\citep{AzSt14a}. As a corollary of Theorem~\ref{th:impossible-group}, we get that when agents may express indifference, there exists no randomized social choice rule that is anonymous, neutral, $\sd$-efficient, and  weak $\sd$ group-strategyproof. 
This proves a weaker version of the conjecture that there exists no randomized social choice rule that is anonymous, neutral, $\sd$-efficient, and weak $\sd$-strategyproof~\citep{ABBH12a}.

%\paragraph{Independence of the axioms}
%We now show that three axioms in Theorem~\ref{th:impossible} as well as the three axioms in Theorem~\ref{th:impossible-group} are independent. 
We now show that if one of $\sd$-efficiency, anonymity, or weak $\sd$-strategyproofness is dropped, then there exist rules that satisfy the other properties mentioned in the two impossibility theorems respectively even for multi-unit demands.
If $\sd$-efficiency is dropped or is replaced by ex post efficiency, then RP satisfies strategyproofness, anonymity, neutrality and ex post efficiency.
If anonymity is dropped, then the priority mechanism achieves $\sd$-efficiency and group $\sd$-strategyproofness.
If weak $\sd$-strategyproofness is dropped, then OPS satisfies the other properties. It remains open whether neutrality is necessarily required to obtain the two impossibility theorems.\\

\section{\PS}

In this section, we examine the properties satisfied by \ps ($MPS$).
Before we proceed, we will try to get a better understanding of how \ps works. \citet{ChKo10a} defined \ps as the rule in which each agent eats his $c$ most preferred objects at speed $1$ during the time interval $t\in [0,1]$. They assumed that at each point each agent has $c$ objects available for consumption during the running of \ps and hence all the objects are consumed at time $1$. We first show that it may be the case that less than $c$ objects are available for consumption. Consider the illustration of \ps in Figure~\ref{figure:muli-ps}. At time $t=7/8$, only $o_4$ is remaining. Hence the first goal is to decide how to define \ps when agents have less than $c$ objects to eat. We resort to the following definition of \ps. 

\begin{quote}
Let $rem(t)$ be the number of objects that have not been completely eaten at time $t$. In multi-unit-eating $PS$, each agent eats his $\min(c,rem(t))$ most preferred
available objects with speed $1$ at every time point until all the objects have been consumed. 
\end{quote}

% 
% \begin{example}[Illustration of \ps]
% 	% 	\begin{align*}
% 	% 1:&\quad o_1,o_2,o_3,o_4 \\
% 	% 2:&\quad o_3,o_2,o_4,o_1 
% 	% \end{align*}
% 	
% 	
% \end{example}

 \begin{figure}[h!]
 \centering
 	\[	
 			\begin{tikzpicture}[scale=0.30]
 				\centering
 				\draw[-] (0,-.5) -- (0,9);     
 				\draw[-] (-.5,0) -- (22.5,0);

 				% \draw[-,dashed] (20/3,10) -- (20/3,0);
 				% 			\draw[-,dashed] (20/2,10) -- (20/2,0);
 				% 			\draw[-,dashed] (280/18,10) -- (280/18,0);
 				% 			\draw[-,dashed] (620/36,10) -- (620/36,0);
 				\draw[-] (22.5,9) -- (22.5,0);

\draw[-] (0,3) -- (22.5,3);
\draw[-] (0,6) -- (22.5,6);
\draw[-] (20,0) -- (20,9);

\draw[-] (10,0) -- (10,9);

\draw[-] (0,9) -- (22.5,9);

\draw[-] (15,0) -- (15,9);

\draw[-] (17.5,0) -- (17.5,9);

 					% \draw[->] (0,10) -- (20/3,10);
 					% 				\draw[->] (20/3,8) -- (620/36,8);
 					% 					\draw[->] (0,6) -- (20/2,6);
 					% 					\draw[->] (20/3,4) -- (280/18,4);
 					% 					\draw[->] (280/18,2) -- (20,2);

 						%	\draw (2,-.6) node(c){\small $0.1$};
 						%	\draw (6,-.6) node(c){\small $0.3$};
						
										\draw (0,-.6) node(c){\small $0$};
 							\draw (20/2,-.6) node(c){\small $1/2$};
							
										\draw (22.5,-.6) node(c){\small $9/8$};
							
 						% \draw (280/18,-.27) node(c){$14/18$};
 						% 
 						% 						\draw (620/36,-.27) node(c){$29/36$};
 							\draw (20,-.6) node(c){\small$1$};
							
									\draw (17.5,-.6) node(c){\small$7/8$};

\draw (15,-.6) node(c){\small$3/4$};

 			%	\draw(-.7,15) node(z){\small $10$};
 				\draw(-2,6) node(z){\small Agent $1$};
 				\draw(-2,3) node(z){\small Agent $2$};

\draw(5,6.5) node(z){\small $o_1,o_2$};

\draw(5,3.5) node(z){\small $o_3,o_2$};

\draw(12.5,6.5) node(z){\small $o_1,o_3$};

\draw(12.5,3.5) node(z){\small $o_3,o_4$};

\draw(16.25,6.5) node(z){\small $o_1,o_4$};

\draw(16.25,3.5) node(z){\small $o_1,o_4$};

\draw(18.75,6.5) node(z){\small $o_4$};

\draw(18.75,3.5) node(z){\small $o_4$};

\draw(21.25,6.5) node(z){\small $o_4$};

\draw(21.25,3.5) node(z){\small $o_4$};

 				\end{tikzpicture}
 		\]
	\begin{align*}
		1:&\quad o_1,o_2,o_3,o_4 \\
		2:&\quad o_3,o_2,o_4,o_1 
		\end{align*}
			\[p=\begin{pmatrix}
		3/4&1/2&1/4&1/4\\
	   1/4&1/2&3/4&3/4
		\end{pmatrix}.\]

 \caption{Illustration of \ps with agents eating their preferred objects over time. The eventual assignment is $p$.} 
\label{figure:muli-ps}
 \end{figure}

We will use $MPS$ as the abbreviation for \ps.
Our first observation is that even though agent may not necessarily eat $c$ objects at each point, each agent eats the same number of objects.

\begin{observation}
At each time point, each agent is consuming the same number of objects. All the agents stop eating at exactly the same time. 
\end{observation}
If the number of objects is less than $c$, then we know that only $c'<c$ objects are remaining.
Next, we study properties of \ps. 
The first things to observe is that \ps runs in linear time and results in a unique fractional assignment. We examine various axiomatic properties of \ps. Our main findings are summarized in the following theorem. We will prove these properties in a series of propositions.

\begin{theorem}
	\PS is linear-time, $\sd$ envy-free, weak $\sd$-strategyproof, and unanimous but not ex post efficient. 
	\end{theorem}

\subsection{Fairness}

We first show that \ps satisfies all the notions of fairness defined in the preliminaries. 
It is easy to see that \ps is anonymous and neutral. Next we show that \ps is $\sd$ envy-free. For the proof, we use an extra bit of notation.
For each set $S\subseteq O$, let the characteristic vector of $S$ be $\hat{S}=(x_1,\ldots, x_m)$ where $x_i=1$ if $i\in S$ and $x_i=0$ if $i\notin S$. 

\begin{proposition}
\PS is $\sd$ envy-free.
\end{proposition}
\begin{proof}%[Proof Sketch]
	When \ps is run, if at least one of the $c$ most preferred available objects of some agent $i\in N$ is finished, agent $i$ starts eating the next most preferred $c$ available objects. Also note that when an agent cannot consume more units of an object, then \emph{no} agent can consume more units of the object either.
We will refer to such a time-point as a breakpoint. The breakpoints are $t_1,\ldots, t_l$. Let $p^k$ be the partial assignment at breakpoint $t_k$.
We prove by induction over $k$, the number of breakpoints in the algorithm, that for each agent $i\in N$, his partial allocation $p^k(i) \mathrel{\pref_i^{\sd}} p^k(j)$ for all $j\in N$. 

For the base case $k=1$, we know that $p^1(i) \mathrel{\pref_i^{\sd}} p^1(j)$ for all $j\in N$
since each agent $i$ was consuming his most preferred $c$ objects. Now let us assume that $p^k(i) \mathrel{\pref_i^{\sd}} p^k(j)$. We show that $p^{k+1}(i) \mathrel{\pref_i^{\sd}} p^{k+1}(j)$. At time $t^k$, let the number of objects that have not been completely even be $c'\leq c$.
Let us consider the time point $t_k+\delta$ for some arbitrarily small $\delta>0$.
From time point $t_k$ to $t_k+\delta$ each agent $i$ consumes $\delta$ amount of $c'$ most preferred objects of $S\subset O$ for which $\delta$ amount is still available. 
Thus $p^k(i)$ is changed to $p^k(i) + \delta(\hat{S})$. In the meanwhile for each $j$, $p(j)$ is changed to $p^k(j) + \delta(\hat{S'})$ where $S'$ consists of $c'$ most preferred objects for which $\delta$ amount is still available. 
Hence, $p^{k+1}(i) \mathrel{\pref_i^{\sd}} p^{k+1}(j)$ for each $i,j\in N$.
	\qed
\end{proof}
\begin{corollary}\PS is weak $\sd$ envy-free. Moreover, 
for the assignment problem without multi-unit demands, $PS$ is $\sd$ envy-free. 

\end{corollary}

\subsection{Strategyproofness}

In this subsection, we examine the strategic aspects of \ps. We show that \ps satisfies $DL$-strategyproofness and hence weak $\sd$-strategyproofness. A random assignment function $f$ is $DL$-strategyproof if $f(\pref)(i)\succsim_i^{DL} f(\pref_i',\pref_{-i})(i)$ for all $\pref_i'\in \mathcal{R}(O)$ and $\pref_i'\in {\mathcal{R}(O)}^{n-1}$.

\begin{lemma}
	$DL$-strategyproofness implies weak $\sd$-strategyproofness.
\end{lemma}
% \begin{proof}
% 	For two allocations $p,q$,
% 	\[p(i)\mathrel{\pref_i^{\sd}}q(i) \implies  p(i)\mathrel{\pref_i^{DL}}q(i).\]
% 	Since $DL$ is a refinement of SD, DL-strategyproofness implies weak SD-strategyproofness. 
% \end{proof}

Next we show that \ps is $DL$-strategyproof. The key to our argument is the insight that an agent cannot get an object with probability one if he does not start eating it from time $t=0$. This contrasts sharply with one-at-a-time $PS$ where an agent can still get an object completely even if he delays eating it.

\begin{lemma}\label{lemma:complete}
An agent cannot get an object $o$ completely if he does not express it as one of his most preferred $c$ objects. 
\end{lemma}
\begin{proof}
	Assume that agent $i$ does not report $o$ as one of his most preferred  $c$ objects but gets it completely. Then while $i$ is eating $o$, there must be at least $c$+1 objects that are still not eaten completely and none of the other agents are eating $o$. Before agent $i$ eats $o$, the number of units eaten by $i$ is at least $1$ and less than $c$. If $i$ has already eaten exactly $c$ units, then it will get zero units of $o$. Now for the $c$ objects it starts eating including $o$, it can eat at most $c-1$ units because it has already eaten at least one unit. Therefore, agent $i$ can eat at most $(c-1)/c$ of $o$. 
		% Assume for contradiction that agent $i$ starts eating $o$ at time $t>0$. Then this means that $i$ finished eating $o$ at time $t>1$. But this means that all agents stopped eating at a time later than $1$. If all agents had been eating their $c$ most preferred available objects then the objects would have been finished at time $1$. This means that at some time point before $t=1$, there were less than $c$ available objects. This implies that other agents also consumed some fraction of object $o$. Hence agent $i$ does not get $o$ completely. Hence an agent can only get an object completely if he starts eating it at time $0$.
	\qed
\end{proof}

\begin{proposition}\label{prop:dlsp}
\PS is $DL$-strategyproof.
\end{proposition}
\begin{proof}%[Proof Sketch]
	
We show that for each agent $i\in N$,
$\text{$MPS$}(N,O, (\pref_i,\pref_{-i}))(i) \pref_i^{DL} \text{$MPS$}(N,O, (\pref_i',\pref_{-i}))(i)$ for all other preferences $\pref_i' \in \mathcal{R}(O)$ and $\pref_{-i}\in {\mathcal{R}(O)}^{n-1}$.
If agent $i$ misreports but eats the same objects at each time point, then $i$ gets exactly the same allocation. Therefore, it is sufficient to show that $i$ gets a less preferred allocation with respect to $DL$ if he does not eat the most preferred available objects at each point. Consider the untruthful report $\succ_i'$ under which at some breakpoint $t$, agent $i$ eats a different set of $\min(c,rem(t))$ objects than when he reports $\succ_i$. Consider  the most preferred object $o$ that $i$ started eating at time $t$ when he was reports $\succ_i$ but does not eat when he reports $\succ_i$. This means that for all $o'\succ_i o$, agent $i$ gets exactly the same units of $o'$ when he reports $\succ_i$ or when he reports $\succ_i'$. 
Since $i$ does not eat $o$ at time $t$ when he reports $\succ_i'$, he eats it at a time later than $t$.
We can assume that $rem(t)>c$ or else agent $i$ will eat the same objects after time $t$ whether he reports $\succ_i$ or $\succ_i'$. We show that $i$ gets strictly less fraction of $o$ when he reports $\succ_i'$. 
We distinguish between two cases: $(1)$ when $i$ eats $o$ when he reports $\succ_i'$, there is at least one other agent $j$ that also eats $o$ at some point.
$(2)$ when $i$ eats $o$ when he reports $\succ_i'$, there is at least one other agent $j$ that also eats $o$ at some point.
In case of $(1)$, $o'$ is in demand and $i$ could have eaten a bigger portion of $o$ had he started eating it earlier such as time $t$. 
In case of $(2)$, no agent started eating $o$ at \emph{any} time point when $i$ reports $\succ_i'$. This implies that $i$ gets $o$ completely. But this is a contradiction because we proved in Lemma~\ref{lemma:complete} that if an agent does not start eating an object at time $0$, then he cannot eat it completely. \qed
\end{proof}

The proposition implies that \PS is weak $\sd$-SP.
As a corollary we also get that for $m=n$, the original $PS$  is weak $\sd$-strategyproof. Our proof simplifies the argument in \citep[Step 2, Proposition 1, ][]{BoMo01a}.

Note that Proposition~\ref{prop:dlsp} crucially depends on the fact that in MPS, each agent tries to eat his $c$ most preferred objects. If each agent eats $c-1$ most preferred objects, then we already know from \citep{Koji09a}, that the rule is then not even weak $\sd$-strategyproof. We note that in contrast to \PS, $OPS$ is not DL-strategyproof and in fact there exists a polynomial-time algorithm for computing a DL best response~\citep{AGM+15c}.

% \begin{proposition}\label{prop:not-sdsp}
% \PS is not $\sd$-strategyproof.
% \end{proposition}
% \begin{proof}
% 
% \end{proof}
% 
% 
% 
% \begin{proposition}\label{prop:not-sdsp}
% For the case of two agents, \PS is strategyproof.
% \end{proposition}
% \begin{proof}
% 
% \end{proof}

%
%
%
%The only way an agent can get a DL improvement is if $i$ gets a highly preferred object with probability one, misreports and delays eating that object and still gets that object with probability one. However, $i$, eats his most preferred $c$ objects at speed one during time period $[0,1]$. Hence if an agent delays eating an object, then it can never eat it completely. 

% PS does not $\sd$-dominates the $RP$ output~\citep{BoMo01a}.
% We saw that as the number of agents is more than the number of agents, one cannot satisfy weak $\sd$-strategyproofness and $\sd$-efficiency if one insists on $\sd$-envy-freeness. 

% One can check whether a random assignment is ex post efficient by using Birkhoff's algorithm. Every time a permutation matrix is used we restrict ourselves to Pareto optimal discrete assignments.

% to check 
% whether a random assignment can be represented by a convex combination of discrete assignments

% \begin{proposition}
% 	
% \end{proposition}

\subsection{Efficiency}

We now consider efficiency of \ps. We first observe that \ps satisfies unanimity.

\begin{proposition}
\PS satisfies unanimity.
\end{proposition}
\begin{proof}
A preference profile admits a perfect assignment only if each agent can get his most preferred $c$ objects. This implies that for any two agents, their sets of $c$ most preferred objects don't intersect. Given this condition, \ps will assign each agent with his most preferred $c$ objects.
	\qed
\end{proof}

Although unanimity is a very undemanding efficiency property, not all assignment rules satisfy unanimity. For example, the uniform rule does not satisfy it. Even if \ps is modified slightly so that agents eat their $c+1$ most preferred objects at the same rate, then the modified rule would not satisfy unanimity. We also note that the allocation of each agent via \ps is $\sd$-preferred over the uniform allocation.

% 
% \begin{lemma}
% 	Discrete $\sd$-efficiency implies ex post efficiency.
% \end{lemma}
% \begin{proof}
% 	Assume $p$ is not ex post efficient. Then  
% \end{proof}

% \begin{conjecture}
% 	A random assignment is ex post efficient  if and only if it is not $\sd$-dominated by a Pareto optimal discrete assignment.
% \end{conjecture}
% \begin{proof}
% 	If $p$ is not $\sd$-dominated by a discrete assignment 
% \end{proof}

\begin{proposition}
For each agent $i\in N$, $i$ $\sd$-prefers his allocation returned by 	
	\ps to the uniform allocation.
\end{proposition}

Informally, an agent gets his worst possible assignment if all the other agents have the same preferences. Even in this case, each agent gets a uniform allocation. Although, \ps satisfies unanimity, an assignment returned by \ps can be represented as a convex combination of Pareto dominated discrete assignments.

% \begin{proof}
% 
% \end{proof}

\begin{proposition}
There exists a preference profile for which the outcome of \ps can be represented as a probability distribution over Pareto dominated discrete assignments.  
\end{proposition}
\begin{proof}
Consider two agents having the following preferences.
\begin{align*}
1:&\quad o_1,o_2,o_3,o_4 \\
2:&\quad o_2,o_1,o_4,o_3 
\end{align*}
The random assignment as a result of \ps is 
\[\begin{pmatrix}
	1/2&1/2&1/2&1/2\\
   1/2&1/2&1/2&1/2
	\end{pmatrix}\]

which can be represented by a probability distribution over the following discrete assignments. 
\[\frac{1}{2}\begin{pmatrix}
	1&0&0&1\\
   0&1&1&0
	\end{pmatrix} + \frac{1}{2}\begin{pmatrix}
	0&1&1&0\\
   1&0&0&1
	\end{pmatrix}. \]
	
	% The random assignment is $\sd$ dominated by 
	% \[\begin{pmatrix}
	% 	1&0&1&0\\
	%    0&1&0&1
	% 	\end{pmatrix}.
	% \]
	
	It can be shown that both discrete assignments are not $\sd$-efficient.
\qed
\end{proof}
\begin{corollary}
\PS is not $\sd$-efficient.
\end{corollary}
\begin{proof}
	An $\sd$-efficient assignment cannot be represented as a convex combination of discrete assignment in which at least one of the assignments is not $\sd$-efficient. If this were the case, then the random assignment is not $\sd$-efficient.
	\qed
\end{proof}

Although the lack of $\sd$-efficiency of \ps  was commented on in the original paper of \citet{ChKo10a}, we show that \ps is surprisingly not even ex post efficient.

% \begin{proposition}
% \PS is not ex post efficient.
% \end{proposition}
% \begin{proof}
% 	\begin{align*}
% 1:&\quad o_1,o_2,o_3,o_4 \\
% 2:&\quad o_3,o_2,o_4,o_1 
% \end{align*}
% 
% 
% The only Pareto optimal assignments are
% 
% \[\begin{pmatrix}
% 	1&1&0&0\\
%    0&0&1&1
% 	\end{pmatrix}\]
% 
% and 
% 
% \[\begin{pmatrix}
% 	1&0&0&1\\
%    0&1&1&0
% 	\end{pmatrix}.\]
% 	
% 	
% 	The outcome of \ps is 
% 	
% 	\[p=\begin{pmatrix}
% 	1&1/2&1/4&1/4\\
%    0&1/2&3/4&3/4
% 	\end{pmatrix}.\]
% 
% 
% Now if random assignment $p$ is ex post efficient, then 
% 
% 
% \[\begin{pmatrix}
% 	1&1/2&1/4&1/4\\
%    0&1/2&3/4&3/4
% 	\end{pmatrix}= \lambda \begin{pmatrix}
% 	1&1&0&0\\
%    0&0&1&1
% 	\end{pmatrix}   + (1-\lambda)\begin{pmatrix}
% 	1&0&0&1\\
%    0&1&1&0
% 	\end{pmatrix}\]
% 	
% 	Then $(1-\lambda)=1/4$ which means $\lambda=3/4$.  
% 	But $1/2=\lambda$ which is a contradiction. 
% 
% 
% \end{proof}

\begin{proposition}
\PS is not ex post efficient even if we allow convex combinations of all deterministic assignments including unbalanced deterministic assignments. 
\end{proposition}
% \begin{proof}[Proof Sketch]
% 	Consider two agents having the following preferences.
% 		\begin{align*}
% 1:&\quad o_1,o_2,o_3,o_4 \\
% 2:&\quad o_3,o_2,o_4,o_1 
% \end{align*}
% 
% 
% The only Pareto optimal feasible assignments are
% 
% \[\begin{pmatrix}
% 	1&1&0&0\\
%    0&0&1&1
% 	\end{pmatrix} \text{ and }
% \begin{pmatrix}
% 	1&0&0&1\\
%    0&1&1&0
% 	\end{pmatrix}.\]
% 	
% 		%If we also allow Pareto optimal assignment in which agents may not get two object then the discrete assignments in which agent $1$ gets $\{a\}$ or $\{a,b,c\}$ or $\{a,b,d\}$ is also Pareto optimal.  Notice that in each Pareto optimal assignment agent $1$ gets $a$ completely. 
% 	%
% 	
% 	The outcome of \ps is 
% 	
% 	\[p=\begin{pmatrix}
% 	7/8&4/8&2/8&3/8\\
%    1/8&4/8&6/8&5/8
% 	\end{pmatrix}.\]
% 	
% 	
% 
% 
% 
% Now if random assignment $p$ is ex post efficient, then it can be expressed as a convex combination of Pareto optimal feasible assignments. Hence, 
%  \begin{align*}
% p=& \lambda \begin{pmatrix}
% 	1&1&0&0\\
%    0&0&1&1
% 	\end{pmatrix}   + (1-\lambda)\begin{pmatrix}
% 	1&0&0&1\\
%    0&1&1&0
% 	\end{pmatrix}.
% \end{align*}
% 	
% 	%Then $(1-\lambda)=3/8$ which means $\lambda=5/8$.  
% 	%But $1=$ which is a contradiction. 
% 
% 
% Since agent $1$ gets $o_1$ completely in each Pareto optimal assignment, $p_{1o_1}$ should be one and not $7/8$. 
% %$7/8=1$ which is a contradiction.
% \end{proof}

\begin{proof}
	Consider two agents having the following preferences.
		\begin{align*}
1:&\quad o_1,o_2,o_3,o_4 \\
2:&\quad o_3,o_2,o_4,o_1 
\end{align*}

A discrete assignment is not $\sd$-efficient if agent $1$ gets $o_3$ or $o_4$ and agent $2$ gets $o_1$. 
The only $\sd$-efficient discrete assignments are
$\begin{pmatrix}
	1&1&0&0\\
   0&0&1&1
	\end{pmatrix} \text{,}
\begin{pmatrix}
	1&0&0&1\\
   0&1&1&0
	\end{pmatrix}, 
	\begin{pmatrix}
		1&0&0&0\\
	   0&1&1&1
		\end{pmatrix},
		\begin{pmatrix}
			1&1&1&0\\
		   0&0&0&1
			\end{pmatrix}
	$, $			\begin{pmatrix}
					1&1&1&1\\
				   0&0&0&0
					\end{pmatrix} \text{ and}
					\begin{pmatrix}
						0&0&0&0\\
					   1&1&1&1
						\end{pmatrix}.$ We note that the outcome of \ps is 
		$p=\begin{pmatrix}
	7/8&4/8&2/8&3/8\\
   1/8&4/8&6/8&5/8
	\end{pmatrix}.$
		%If we also allow Pareto optimal assignment in which agents may not get two object then the discrete assignments in which agent $1$ gets $\{a\}$ or $\{a,b,c\}$ or $\{a,b,d\}$ is also Pareto optimal.  Notice that in each Pareto optimal assignment agent $1$ gets $a$ completely. 
	%
	% Note that agent $1$ gets $o_1$ completely in each Pareto optimal assignment. 
Now if random assignment $p$ is ex post efficient, then it can be expressed as a convex combination of $\sd$-efficient feasible discrete assignments. Since $p(2)(o_1)>0$, this is only possible if $	\begin{pmatrix}
		0&0&0&0\\
	   1&1&1&1
		\end{pmatrix}$ is used in the convex combination. But since agent $2$ does not get $o_1$ in any other discrete permutation, this means that if any convex combination of $\sd$-efficient discrete assignments is used to obtain $p$, then in each discrete $\sd$-efficient assignment used the following three cases can occur: $(i)$ $2$ gets both $o_2$ and $o_1$; $(ii)$ $2$ gets neither $o_2$ nor $o_1$ and $(iii)$ $2$ gets $o_2$ but not $o_1$. Hence, 
it must be that $p(2)(o_2)\geq p(2)(o_1)$. But this is a contradiction.
	\qed
\end{proof}

% \textbf{page 15, the proof of Proposition 6, among six Pareto optimal discrete assignments, consider the last four assignments, at which an agent receives more than two objects. Are they feasible in this model? The violation of ex post efficiency should be shown within the range of feasible discrete allotments. I don’t think it is meaningful otherwise.}

\section{Conclusions}

\begin{table*}[h!]
%	\small
	\centering
		\scalebox{1}{
\centering
%\scriptsize 

\begin{tabular}{lccccc}
\toprule
%&&&Random&One-at-a-time&Multi-unit-eating\\ 
&Uniform&Priority&$RP$&$OPS$&$MPS$\\ \midrule
$\sd$-efficiency&-&+&-&+&-\\ 
%discrete $\sd$-efficiency&-&+&-&+&-\\ 
ex post efficient&-&+&+&+&-\\ 
%non-wastefulness&-&+&+&+&+\\ 
%discrete assignment is PO&-&+&+&+&+\\ 
unanimity&-&+&+&+&+\\ 
\midrule
$\sd$ envy-freeness&+&-&-&+&+\\ 
weak $\sd$ envy-freeness&+&-&+&+&+\\ 
anonymous&+&-&+&+&+\\
neutrality&+&+&+&+&+\\
\midrule
$\sd$-SP&+&+&+&-&-\\
%SD-SP for 2 agents&+&+&+&-&+\\
$DL$-SP&+&+&+&-&+\\
weak $\sd$-SP&+&+&+&-&+\\%  
\midrule
 polynomial-time &+&+&-&+&+\\
% consistency&+&+&-&+&+\\
\bottomrule
\end{tabular}
}
\caption{Assignment rules for allocating multiple objects to agents with strict preferences. Most of the properties of rules other than $MPS$ are stated in \cite{Koji09a}.
}
\label{table:summarymulti}
\end{table*}

In this paper, we showed a general impossibility result concerning randomized assignment with multi-unit demands. % The statement applies to many resource allocation settings.
Another impossibility result requiring weak $\sd$-group-strategyproofness applies to
randomized assignment without multi-unit demands.  As a corollary of the second impossibility, we also obtain the corresponding impossibility in the domain of randomized voting.

We then presented a definition of \ps. \PS has previously only been defined inaccurately in the literature. We showed that whereas \ps satisfies some compelling fairness and strategic properties, it does not satisfy reasonable efficiency requirements. We note that the positive results of \PS even hold if $m$ is not a multiple of $n$. In this case, agents eat a maximum of $\ceil{m/n}$ houses at any time.

Our findings concerning \ps are summarized in Table~\ref{table:summarymulti} which also provides a comparison with other random assignment rules. In view of the impossibility result (Theorem~\ref{th:impossible}), it is not possible to achieve the desirable properties of $PS$ and \ps simultaneously.  It is easy to see that the choice of an assignment rule depends on which properties are prioritized. Our paper helps clarify the relative merits of various randomized assignments rules. 
It is an open problem whether ex post efficiency, weak $\sd$-strategyproofness and $\sd$ envy-freeness are compatible in the multi-unit case.
 %(see Figure~\ref{fig:ops-mps}).
We leave a characterization of \ps for future work.  
%Another interesting problem is the computational complexity of manipulation under \ps.

\section*{Acknowledgments}

This material is based upon work supported by
the Australian Government's
Department of Broadband, Communications and the Digital
Economy, the Australian Research Council, the Asian
Office of Aerospace Research and Development through
grant AOARD-124056.

 %\bibliography{../../pamas/abb,../../pamas/pamas,../../pamas/brandt,../../pamas/aziz}

\begin{thebibliography}{34}
\providecommand{\natexlab}[1]{#1}
\providecommand{\url}[1]{\texttt{#1}}
\expandafter\ifx\csname urlstyle\endcsname\relax
  \providecommand{\doi}[1]{doi: #1}\else
  \providecommand{\doi}{doi: \begingroup \urlstyle{rm}\Url}\fi

\bibitem[Abdulkadiro{\u{g}}lu and S{\"o}nmez(1998)]{AbSo98a}
A.~Abdulkadiro{\u{g}}lu and T.~S{\"o}nmez.
\newblock Random serial dictatorship and the core from random endowments in
  house allocation problems.
\newblock \emph{Econometrica}, 66\penalty0 (3):\penalty0 689--702, 1998.

\bibitem[Abdulkadiro{\u{g}}lu and S{\"o}nmez(1999)]{AbSo99a}
A.~Abdulkadiro{\u{g}}lu and T.~S{\"o}nmez.
\newblock House allocation with existing tenants.
\newblock \emph{Journal of Economic Theory}, 88\penalty0 (2):\penalty0
  233--260, 1999.

\bibitem[Abraham et~al.(2005)Abraham, Cechl{\'a}rov{\'a}, Manlove, and
  Mehlhorn]{ACMM05a}
D.~J. Abraham, K.~Cechl{\'a}rov{\'a}, D.~Manlove, and K.~Mehlhorn.
\newblock Pareto optimality in house allocation problems.
\newblock In \emph{Proceedings of the 16th International Symposium on
  Algorithms and Computation (ISAAC)}, volume 3341 of \emph{Lecture Notes in
  Computer Science (LNCS)}, pages 1163--1175, 2005.

\bibitem[Athanassoglou and Sethuraman(2011)]{AtSe11a}
S.~Athanassoglou and J.~Sethuraman.
\newblock House allocation with fractional endowments.
\newblock \emph{International Journal of Game Theory}, 40\penalty0
  (3):\penalty0 481--513, 2011.

\bibitem[Aziz and Stursberg(2014)]{AzSt14a}
H.~Aziz and P.~Stursberg.
\newblock A generalization of probabilistic serial to randomized social choice.
\newblock In \emph{Proceedings of the 28th AAAI Conference on Artificial
  Intelligence (AAAI)}, pages 559--565. AAAI Press, 2014.

\bibitem[Aziz et~al.(2013{\natexlab{a}})Aziz, Brandt, and Brill]{ABB13b}
H.~Aziz, F.~Brandt, and M.~Brill.
\newblock The computational complexity of random serial dictatorship.
\newblock \emph{Economics Letters}, 121\penalty0 (3):\penalty0 341--345,
  2013{\natexlab{a}}.

\bibitem[Aziz et~al.(2013{\natexlab{b}})Aziz, Brandt, and Brill]{ABBH12a}
H.~Aziz, F.~Brandt, and M.~Brill.
\newblock On the tradeoff between economic efficiency and strategyproofness in
  randomized social choice.
\newblock In \emph{Proceedings of the 12th International Conference on
  Autonomous Agents and Multi-Agent Systems (AAMAS)}, pages 455--462. IFAAMAS,
  2013{\natexlab{b}}.

\bibitem[Aziz et~al.(2013{\natexlab{c}})Aziz, Brandt, and Stursberg]{ABS13a}
H.~Aziz, F.~Brandt, and P.~Stursberg.
\newblock On popular random assignments.
\newblock In \emph{Proceedings of the 6th International Symposium on
  Algorithmic Game Theory (SAGT)}, volume 8146 of \emph{Lecture Notes in
  Computer Science (LNCS)}, pages 183--194. Springer-Verlag,
  2013{\natexlab{c}}.

\bibitem[Aziz et~al.(2015)Aziz, Gaspers, Mackenzie, Mattei, Narodytska, and
  Walsh]{AGM+15c}
H.~Aziz, S.~Gaspers, S.~Mackenzie, N.~Mattei, N.~Narodytska, and T.~Walsh.
\newblock Manipulating the probabilistic serial rule.
\newblock In \emph{Proceedings of the 14th International Conference on
  Autonomous Agents and Multi-Agent Systems (AAMAS)}, pages 1451--1459, 2015.

\bibitem[Bhalgat et~al.(2011)Bhalgat, Chakrabarty, and Khanna]{BCK11a}
A.~Bhalgat, D.~Chakrabarty, and S.~Khanna.
\newblock Social welfare in one-sided matching markets without money.
\newblock In \emph{Proceedings of APPROX-RANDOM}, pages 87--98, 2011.

\bibitem[Bogomolnaia and Moulin(2001)]{BoMo01a}
A.~Bogomolnaia and H.~Moulin.
\newblock A new solution to the random assignment problem.
\newblock \emph{Journal of Economic Theory}, 100\penalty0 (2):\penalty0
  295--328, 2001.

\bibitem[Bouveret and Lang(2011)]{BoLa11a}
S.~Bouveret and J.~Lang.
\newblock A general elicitation-free protocol for allocating indivisible goods.
\newblock In \emph{Proceedings of the 22nd International Joint Conference on
  Artificial Intelligence (IJCAI)}, pages 73--78. AAAI Press, 2011.

\bibitem[Bouveret et~al.(2010)Bouveret, Endriss, and Lang]{BEL10a}
S.~Bouveret, U.~Endriss, and J.~Lang.
\newblock Fair division under ordinal preferences: Computing envy-free
  allocations of indivisible goods.
\newblock In \emph{Proceedings of the 19th European Conference on Artificial
  Intelligence (ECAI)}, pages 387--392, 2010.

\bibitem[Budish(2011)]{Budi11a}
E.~Budish.
\newblock The combinatorial assignment problem: Approximate competitive
  equilibrium from equal incomes.
\newblock \emph{Journal of Political Economy}, 119\penalty0 (6):\penalty0
  1061--1103, 2011.

\bibitem[Budish et~al.(2013)Budish, Che, Kojima, and Milgrom]{BCKM12a}
E.~Budish, Y.-K. Che, F.~Kojima, and P.~Milgrom.
\newblock Designing random allocation mechanisms: {T}heory and applications.
\newblock \emph{American Economic Review}, 103\penalty0 (2):\penalty0 585--623,
  2013.

\bibitem[Chambers(2004)]{Cham04a}
C.~Chambers.
\newblock Consistency in the probabilistic assignment model.
\newblock \emph{Journal of Mathematical Economics}, 40:\penalty0 953--962,
  2004.

\bibitem[Che and Kojima(2010)]{ChKo10a}
Y.-K. Che and F.~Kojima.
\newblock Asymptotic equivalence of probabilistic serial and random priority
  mechanisms.
\newblock \emph{Econometrica}, 78\penalty0 (5):\penalty0 1625---1672, 2010.

\bibitem[Cr{\`e}s and Moulin(2001)]{CrMo01a}
H.~Cr{\`e}s and H.~Moulin.
\newblock Scheduling with opting out: Improving upon random priority.
\newblock \emph{Operations Research}, 49\penalty0 (4):\penalty0 565--577, 2001.

\bibitem[Ehlers and Klaus(2003)]{EhKl03a}
L.~Ehlers and B.~Klaus.
\newblock Probabilistic assignments of identical indivisible objects and
  probabilistic uniform rules.
\newblock \emph{Review of Economic Design}, 8:\penalty0 249--268, 2003.

\bibitem[G{\"a}rdenfors(1973)]{Gard73b}
P.~G{\"a}rdenfors.
\newblock Assignment problem based on ordinal preferences.
\newblock \emph{Management Science}, 20:\penalty0 331--340, 1973.

\bibitem[Guo and Conitzer(2010)]{GuCo10a}
M.~Guo and V.~Conitzer.
\newblock Strategy-proof allocation of multiple items without payments or
  priors.
\newblock In \emph{Proceedings of the 9th International Conference on
  Autonomous Agents and Multi-Agent Systems (AAMAS)}, pages 881--888. IFAAMAS,
  2010.

\bibitem[Hashimoto(2013)]{Hash13a}
T.~Hashimoto.
\newblock The generalized random priority mechanism with budgets.
\newblock 2013.

\bibitem[Hatfield(2009)]{Hatf09a}
J.~W. Hatfield.
\newblock Strategy-proof, efficient, and nonbossy quota allocations.
\newblock \emph{Social Choice and Welfare}, 33:\penalty0 505---515, 2009.

\bibitem[Heo(2011)]{Heo11a}
E.~J. Heo.
\newblock Probabilistic assignment with multiple demands: A generalization and
  a characterization of the serial rule.
\newblock Technical Report 1809195, SSRN, 2011.

\bibitem[Hylland and Zeckhauser(1979)]{HyZe79a}
A.~Hylland and R.~Zeckhauser.
\newblock The efficient allocation of individuals to positions.
\newblock \emph{The Journal of Political Economy}, 87\penalty0 (2):\penalty0
  293--314, 1979.

\bibitem[Kalinowski et~al.(2013)Kalinowski, Narodytska, Walsh, and
  Xia]{KNWX13a}
T.~Kalinowski, N.~Narodytska, T.~Walsh, and L.~Xia.
\newblock Strategic behavior when allocating indivisible goods sequentially.
\newblock In \emph{Proceedings of the 27th AAAI Conference on Artificial
  Intelligence (AAAI)}, pages 452--458. AAAI Press, 2013.

\bibitem[Katta and Sethuraman(2006)]{KaSe06a}
A-K. Katta and J.~Sethuraman.
\newblock A solution to the random assignment problem on the full preference
  domain.
\newblock \emph{Journal of Economic Theory}, 131\penalty0 (1):\penalty0
  231--250, 2006.

\bibitem[Kojima(2009)]{Koji09a}
F.~Kojima.
\newblock Random assignment of multiple indivisible objects.
\newblock \emph{Mathematical Social Sciences}, 57\penalty0 (1):\penalty0
  134---142, 2009.

\bibitem[Nguyen et~al.(2015)Nguyen, Peivandi, and Vohra]{NPV15a}
T.~Nguyen, A.~Peivandi, and R.~Vohra.
\newblock Assignment problems with complementarities.
\newblock 2015.

\bibitem[Svensson(1994)]{Sven94a}
L.-G. Svensson.
\newblock Queue allocation of indivisible goods.
\newblock \emph{Social Choice and Welfare}, 11:\penalty0 323--330, 1994.

\bibitem[Svensson(1999)]{Sven99a}
L-G Svensson.
\newblock Strategy-proof allocation of indivisible goods.
\newblock \emph{Social Choice and Welfare}, 16\penalty0 (4):\penalty0 557--567,
  1999.

\bibitem[Wilson(1977)]{Wils77a}
L.~Wilson.
\newblock Assignment using choice lists.
\newblock \emph{Operations Research Quarterly}, 28\penalty0 (3):\penalty0
  569---578, 1977.

\bibitem[Yilmaz(2009)]{Yilm09a}
O.~Yilmaz.
\newblock Random assignment under weak preferences.
\newblock \emph{Games and Economic Behavior}, 66\penalty0 (1):\penalty0
  546--558, 2009.

\bibitem[Young(1995)]{Youn95b}
H.~P. Young.
\newblock Dividing the indivisible.
\newblock \emph{American Behavioral Scientist}, 38:\penalty0 904--920, 1995.

\end{thebibliography}

\end{document}